\documentclass[conference]{IEEEtran}
\IEEEoverridecommandlockouts

\usepackage{cite}
\usepackage{amsmath,amssymb,amsfonts}
\usepackage{algorithm}
\usepackage{algorithmic}
\usepackage{graphicx}
\usepackage{textcomp}
\usepackage{xcolor}
\usepackage{amsmath}
\usepackage{graphicx}

\usepackage{amsmath,amsthm,amssymb,scrextend}

\renewcommand{\qed}{\hfill$\blacksquare$}

\renewenvironment{proof}{\begin{addmargin}[1em]{0em}\begin{newproof}}{\end{newproof}\end{addmargin}\qed}

\newenvironment{lemma}[2][Lemma]{\begin{trivlist}
\item[\hskip \labelsep {\bfseries #1}\hskip \labelsep {\bfseries #2.}]}{\end{trivlist}}

\begin{document}

\title{Verifying Shortest Paths in Linear Time$^*$\\

\thanks{* This paper is based upon work supported by Science, Technology, \& Innovation Funding Authority (STDF)
under grant number 45542.\\
Email: Prof. Amr Elmasry --
amr.elmasry@eui.edu.eg. 
}
}

\author{
    \IEEEauthorblockN{Ahmed Shokry\textsuperscript{*}, Amr Elmasry\textsuperscript{*,\textdagger}, Ayman Khalafallah\textsuperscript{*}, Amr Aly\textsuperscript{*}}
    \IEEEauthorblockA{\textsuperscript{*}Computer and Systems Engineering Department, Alexandria University, Egypt}
    \IEEEauthorblockA{\textsuperscript{\textdagger}Egypt University of Informatics, Egypt}
}

\makeatletter
\def\@fnsymbol#1{\ifcase#1\or *\or \dagger\or \ddagger\or \mathsection\or \mathparagraph\or \|\or **\or \dagger\dagger\or \ddagger\ddagger\else\@ctrerr\fi}
\makeatother


\maketitle

\begin{abstract}
In this paper we propose a linear-time certifying algorithm for the single-source shortest-path problem capable of verifying graphs with positive, negative, and zero arc weights. Previously proposed linear-time approaches only work for graphs with positive arc weights. 
\end{abstract}


\section{Introduction}
The shortest-path problem is one of the most important problems in combinatorial optimization and graph theory. It plays a critical role in a wide array of practical applications~\cite{goldberg2005computing, fu2006heuristic, panahi2008gis, soueres1996shortest, aridhi2015mapreduce}. At its core, the problem involves finding the shortest path between two vertices in a graph, which can model a lot of today's applications (e.g., road networks). Variations of the shortest-path problem include: the single-source shortest-path problem, which seeks the shortest paths from a given source vertex to all other vertices in the graph, and the all-pairs shortest-path problem, which aims to find the shortest paths between every pair of vertices~\cite{cormen2022introduction}. 
%


The single-source shortest-path problem is defined by $(G,s,l)$, where $G= (V,A)$ is a directed weighted graph, $V$
is the set of $n$ vertices, $A$ is the set of $m$ arcs, $s$ is the source vertex, and $l : A \to \mathbb{R}$ is a length function, where $l(u,v)$ is the length of the arc $(u,v)$. The shortest path is a path of arcs with the minimum total length. The shortest path is undefined if $G$ has a cycle with negative total length. 
The well-known Bellman-Ford algorithm~\cite{Bellman1958} requires $O(n \cdot m)$ time to construct the shortest-path tree from a given source vertex for a graph that may have negative weights.
 
Given a shortest-path algorithm that generates, along with each output, a certificate (witness) that confirms the path's optimality, a certifying shortest-path algorithm is important to provide a fast proof of correctness, ensuring that the computed path is indeed optimal. Certifying shortest-path algorithms are important in critical applications such as navigation systems, emergency response, and network security~\cite{mcconnell2011certifying}. In addition, such algorithms provide a benchmark for developing and testing new shortest-path algorithms.

The current linear-time certifying algorithm addresses the single-source shortest-path problem for graphs with only positive arc weights~\cite{mcconnell2011certifying}, which simplifies the problem and avoids complications such as negative-weight cycles. Still, many real-world applications, such as in financial networks or transportation systems with bidirectional cost adjustments, require handling graphs with negative or zero arc weights.


In this paper we propose a linear-time certifying algorithm for the single-source shortest-path problem that works with graphs containing positive, negative, and zero arc weights. Our algorithm runs in $O(m)$ time, offering an $O(n)$ speedup over recomputing the shortest paths and enabling efficient benchmarking for the general single-source shortest-path problem.

\section{The certifying algorithm}

To certify the shortest-path problem, the prover provides the claimed shortest-path distance $D[v]$ from the source vertex $s$ to every vertex $v$. The verifier would then determine, for every vertex $v$, whether the value $D[v]$ matches the true shortest-path distance $d[v]$, without recomputing $d[v]$.

\subsection{Constraints}
To prove that $\forall v\in V$, $D[v]$ = $d[v]$, we need to verify four sets of constraints. 

At first we need to check that the shortest path distance to the source $s$ is 0.
\begin{equation}\label{start_ct}
    D[s] = 0
\end{equation}
This constraint ensures that the distance from the source to itself is correctly initialized.
%
%

Second, we need to check that for all vertices, the shortest path distance can not be reduced. 
\begin{equation}\label{relax_ct}
    \forall (u,v) \in A,\ \ \ D[v] \leq D[u] + l(u, v)
\end{equation}
This constraint is derived from the Bellman-Ford relaxation step, ensuring no vertex $v$ can have a distance 
$D[v]$ less than the distance to any of its predecessors plus the cost of the arc connecting this predecessor to $v$. If $D[v]$ exceeds $D[u]+l(u,v)$, it would imply that a shorter path exists to $v$ through $u$, contradicting the validity of $D[v]$.

%


Then, we need to ensure that every vertex is reachable from the source via a path, such that for every arc $(u,v)$ on these paths $D[v]=D[u]+l(u,v)$.
\begin{align}\label{reachability_ct}
\forall x \in V, &\ \exists \text{ path } P \text{ from } s \text{ to } x \text{ such that } \notag \\
& \forall (u,v) \in P,\ D[v] = D[u] + l(u,v),\ 
\end{align}
These constraints ensure that every vertex is reachable from $s$ through a valid shortest path $P$, where every vertex along the way on $P$ has the prefix sub-path of $P$ as its own shortest path from $s$. This guarantees the span of the shortest-path tree over all vertices, thereby ensuring the validity of the certificate.

Finally, for every vertex $v$ that is not reachable from $s$ in $G$, it should be the case that the prover has set the claimed shortest-path value $D[v]$ to infinity.
\begin{equation}\label{unreachable}
    \forall v \in V \text{ not reachable from } s,\ \ \ D[v] = \infty
\end{equation}
Algorithm~\ref{sssp} summarizes the proposed certifying algorithm.

 \begin{algorithm}[!t]
\caption{ \textsc{Certify}($G$, $s$, $D[v]$)}\label{sssp}
 \begin{algorithmic}[1]
		\IF{$D[s] \neq 0$}
            \STATE \textbf{return} False
    \ENDIF
    \STATE stack.push($s$)
    \STATE mark $s$ as the only reachable vertex

    \WHILE{stack is not empty}
        \STATE $u \gets$ stack.pop()
        \FOR{\texttt{each $v \in u.adj[]$}}
					\IF{$ D[v] > D[u] + l(u,v)$}
            \STATE \textbf{return} False
					\ELSE
						\IF{$v \notin reachable$ \textbf{and} $D[v]=D[u]+l(u, v)$}
                \STATE mark $v$ as reachable
                \STATE stack.push($v$)
						\ENDIF
					\ENDIF	
        \ENDFOR
    \ENDWHILE
    
    \FOR{\texttt{all $v \in V$}}
        \IF{$v \notin reachable$ \textbf{and} $D[v] \neq \infty$}
            \STATE \textbf{return} False
        \ENDIF
    \ENDFOR
\STATE \textbf{return} True	

 \end{algorithmic}
 \end{algorithm}

\subsection{Correctness}

\begin{lemma}{1}
If the graph has a negative cycle, there must exist an arc $(u,v)$ on the cycle where $D[v] > D[u] +l(u,v)$,
and the constraints can not be satisfied.
\end{lemma}

\begin{proof}
If Inequality~\ref{relax_ct} is satisfied allover the arcs of a cycle, summing up the left-hand side and the right-hand side, the prover's distances cancels each other resulting in the sum of the weights of the arcs of the cycle being nonnegative. This means that if the graph has a negative cycle, then this constraint can not be satisfied allover the arcs of the cycle.
\end{proof}

\begin{lemma}{2}
Assume that the algorithm accepts the certificate, then for every vertex $v$ in $V$, the prover's shortest-path distance $D[v]$ will be equal to the actual shortest-path distance $d[v]$. Hence, $D[v] = d[v]$, $\forall v \in V$. 
\end{lemma}
\begin{proof}

First, we prove that $D[v] \leq d[v]$ for every vertex $v \in V$ that is reachable from $s$.
Consider any vertex $v \neq s$ that is reachable from $s$, and let $v_0 = s, v_1,..., v_k = v$ be a shortest path from $s$ to $v$. Then, $d[v_{i+1}] = d[v_i] + l(v_i, v_{i+1})$ for all $i$ where $0 \leq i \leq k$. We prove by induction on $i$ that $D[v_i] \leq d[v_i]$ for all $i$. For $i = 0$,  $D[v_0] = D[s]= 0 = d[v_0]$ from Equation~\ref{start_ct}. Assume that the induction hypothesis is true up to $v_i$ for some $i \geq 0$. Using Inequality~\ref{relax_ct}, 
 \begin{flalign*}
 D[v_{i+1}] &\leq D[v_i] + l(v_i, v_{i+1})\\
  &\leq d[v_i] + l(v_i, v_{i+1}) \\
  &= d[v_{i+1}]
 \end{flalign*}
It follows that the hypothesis holds for $i+1$.

Next, we prove that $D[v] \geq d[v]$ for every vertex $v \in V$ that is reachable from $s$.
Consider any vertex $v \neq s$ that is reachable from $s$, and let $v_0 = s, v_1,..., v_k = v$ be a path along which Equation~\ref{reachability_ct} is satisfied. Then, $D[v_{i+1}] = D[v_i] + l(v_i, v_{i+1})$ for all $i$ where $0 \leq i \leq k$. We prove by induction on $i$ that $d[v_i] \leq D[v_i]$ for all $i$. 
For $i = 0$,  $d[v_0] = d[s]= 0 = D[v_0]$ from Equation~\ref{start_ct}. Assume that the induction hypothesis is true up to $v_i$ for some $i \geq 0$. From the shortest-paths property, 
 \begin{flalign*}
 d[v_{i+1}] &\leq d[v_i] + l(v_i, v_{i+1})\\
  &\leq D[v_i] + l(v_i, v_{i+1}) \\
  &= D[v_{i+1}]
 \end{flalign*}
It follows that the hypothesis holds for $i+1$.

Alternatively, if $v$ is not reachable from $S$, then $d[v]=D[v]=\infty$ by Equation~\ref{reachability_ct}.
\end{proof}

\begin{lemma}{3}
Assume that the algorithm rejects the certificate, then there exists a vertex $v$ in $V$ where $D[v] \neq d[v]$.
\end{lemma}
\begin{proof}
If it turns out that Equation~\ref{start_ct} is not satisfied, then $D[s] \neq d[s]$.

If it turns out that Inequality~\ref{relax_ct} does not hold for an arc $(u,v)$, then $D[v] > D[u] + l(u,v)$.  In this case, if $D[u]$ is indeed the shortest-path value to $u$, meaning that $D[u]=d[u]$, then $D[v]>d[u]+l(u,v)$.
From the shortest-path property, we have $d[v] \leq d[u] + l(u,v)$. It follows that $D[v]>d[u]+l(u,v) \geq d[v]$. In other words, if Inequality~\ref{relax_ct} does not hold for an arc $(u,v)$, then either $D[u] \neq d[u]$ or $D[v] \neq d[v]$.

If it turns out that Equation~\ref{reachability_ct} does not hold for every arc $(u,v)$ entering a given vertex $v$ that is reachable from $s$, assuming Inequality~\ref{relax_ct} holds, then $D[v]< D[u]+l(u,v)$. For every vertex $v$ reachable from $s$ there must exist an arc $(u,v)$ for which $d[v]=d[u]+l(u,v)$. Assume that for this vertex $u$, $D[u]=d[u]$. It follows that $d[v]=D[u]+l(u,v) > D[v]$. In other words, if Equation~\ref{reachability_ct} does not hold for every arc entering vertex $v$, then there is an arc $(u,v)$ where either $D[u] \neq d[u]$ or $D[v] \neq d[v]$.

If it turns out that Equation~\ref{unreachable} is not satisfied at a vertex $v$ that is not reachable from $s$, then $D[v] \neq d[v]$. 
\end{proof}





\subsection{Complexity}
The time bound for certifying the single-source shortest-path problem is indicated in the next lemma.

\begin{lemma}{1}
     The proposed certifying algorithm for the single-source shortest-path problem runs in $O(n+m)$ time.
\end{lemma}

\begin{proof}
The proposed certifying algorithm maintains four sets of constraints. Verifying the first condition requires $O(1)$, while verifying the second one for every arc requires $O(m)$ time. To verify the third constraints, we can use one DFS from $s$ and hence requires $O(n+m)$ time. Verifying the fourth constraint is done afterward by checking the given distances for all the unreachable vertices from $s$ in $(n)$ time.
\end{proof}

\section{Conclusion}
In this paper we presented a linear-time certifying algorithm for the shortest-path problem for graphs with positive, zero and negative arc weights. 
The preexisting shortest-path certifying algorithm is restricted to positive arc lengths only. 

Once we have the certificate, it is definitly better to run an $O(n + m)$ certifying algorithm in comparison with computing the shortest paths from scratch by running the Bellman-Ford or an equivalent algorithm that runs in $O(n \cdot m)$ time.  


\end{document}